\documentclass[a4paper,12pt]{article}
      
\usepackage[T1]{fontenc}
\usepackage[latin1]{inputenc}
\usepackage{latexsym}
\usepackage{amssymb} 
\usepackage{fancyvrb}

\usepackage[T1]{fontenc}
\usepackage[latin1]{inputenc}
\usepackage[english]{babel}
\usepackage{latexsym} 
\usepackage{amssymb}
\usepackage{amsmath}
\usepackage{array}
\usepackage{bbding} 
\usepackage{mathpazo} 
\usepackage[dvips]{graphicx}
\usepackage{amsthm}
\usepackage{xcolor}
\usepackage{setspace}
\usepackage{caption}
\usepackage{subcaption}
\allowdisplaybreaks

\usepackage[top=2.5cm,bottom=2.5cm,right=2.5cm,left=2.5cm,headsep=0.4in,includehead]{geometry}
\usepackage{hyperref}
\theoremstyle{plain}
\newtheorem{theor}{Theorem}[section]

\newtheorem{defi}{Definition}[section]
\newtheorem{exa}{Example}[section]

\theoremstyle{definition}

\theoremstyle{remark}

\newcommand{\reviewtimetoday}[2]{
}

\reviewtimetoday{\today}{Draft Version v.0.1}

\begin{document}

\title{\textbf{Sharing delay costs in stochastic projects}}

\author{J.C. Gon\c{c}alves-Dosantos$^1$, I. Garc\'{\i}a-Jurado, J. Costa}

\date{\small{\emph{Grupo MODES, Departamento de Matem\'aticas, Universidade da Coru\~{n}a,\\ Campus de Elvi\~{n}a, 15071 A Coru\~{n}a}}}

\maketitle
\onehalfspacing

\abstract{An important problem in project management is determining ways to distribute amongst activities the costs that are incurred when a project is delayed because some activities end later than expected. In this study, we address this problem in stochastic projects, where the durations of activities are unknown but their corresponding probability distributions are known. We propose and characterise an allocation rule based on the Shapley value, illustrate its behaviour by using examples, and analyse features of its calculation for large problems.}

\vspace{0.5cm}


\footnotetext[1]{
	Corresponding author. E-mail: juan.carlos.goncalves@udc.es.} 

\section{Introduction}
Project management is a field within operations research that provides managers with techniques to select, plan, execute, and monitor projects. An important issue in project management is time management, which generally call for careful planning of project activities to meet various project delivery dates, especially the final delivery date. Normally, a delay in the final delivery date incurs a cost that is often specified by contract. Sometimes, projects are not developed by one agent but a group of agents. When there is a delay in one of such joint projects, the manner of allocating the delay cost amongst the several participating agents may not be clear. This study deals with the problem of sharing delay costs in a joint project by using cooperative game theory.

In the last few years, several papers have been written proposing and studying allocation rules for delay costs. Berganti\~{n}os and S\'anchez (2002) proposed a rule based on the serial cost-sharing problem. Br\^anzei et al. (2002) provided two rules using, respectively, a game theoretical and bankruptcy-based approach. In Castro et al. (2007), the core of a class of transferable utility cooperative game (in short, a TU-game) arising from a delay cost-sharing problem was studied. Est\'evez-Fern\'andez et al. (2007) and Est\'evez-Fern\'andez (2012) dealt with some classes of TU-games associated with projects whose activities might have been delayed or advanced by generating delay costs or acceleration benefits of the corresponding projects. In Berganti\~{n}os et al. (2018), a consistent rule based on the Shapley value was introduced and analysed. All these papers tackle {\em deterministic scheduling problems with delays}. One such problem is that of a deterministic delayed project. By {\em deterministic project}, we mean a set of activities to be performed with respect to an order of precedence and a description of their estimated durations; by {\em delayed deterministic project}, we mean a project that has been performed, description of the observed durations of the activities according to which the project has lasted longer than expected, and cost function that indicates the delay cost associated with the durations of the activities. 

A natural extension of deterministic problems with delays can be found in {\em stochastic scheduling problems with delays}, which we introduce and analyse in this study. To the best of our knowledge, these problems have not been treated in literature, although Castro et al. (2014) considered the problem of allocating slacks in a stochastic PERT network,\footnote{PERT is the acronym of Program Evaluation and Review Technique, a tool used in project management, first developed by the United States Navy in the 1950s.} which is a related but different problem. Herroelen and Leus (2005) surveyed literature on project management under uncertainty. In a stochastic scheduling problem with delays, the manager has a description of the probability distributions of the random variables modelling the durations of the activities instead of simply their estimated durations. In most cases, managers have information about random variables---for instance, their empirical distributions---based on the durations of similar activities in past projects of the same type. 

The remainder of this paper is organised as follows: In Section 2, we describe the problem and indicate its characteristics in a deterministic setting. In Section 3, we propose a rule based on the Shapley value in this context, study its properties, and provide an axiomatic characterisation based on a balancedness property. We also illustrate the performance of our rule by using two examples. Finally, Section 4 addresses some computational issues related to our rule.

\section{The problem}
In this section, we describe the problem with which we deal. We first formally introduce a deterministic scheduling problem with delays following Berganti\~{n}os et al. (2018):

\begin{defi}
	\label{dsp}
	A deterministic scheduling problem with delays $P$ is a tuple $(N,\prec,x^0,x,C)$ where:
	\begin{itemize}
		\item $N$ is the finite non-empty set of activities. 
		\item $\prec$ is a binary relation over $N$ satisfying asymmetry and transitivity. For every $i,j\in N$, we interpret $i\prec j$ as ``activity $j$ cannot start until activity $i$ has finished".
		\item $x^0\in\mathbb{R}^N$ is the vector of planned durations. For every $i\in N$, $x^0_i$ is a non-negative real number indicating the planned duration of activity $i$.
		\item $x\in\mathbb{R}^N$ is the vector of actual durations. For every $i\in N$, $x_i\geq x_i^0$ indicates the duration of activity $i$.
		\item $C:\mathbb{R}^N\rightarrow \mathbb{R}$ is the delay cost function. We assume that $C$ is non-decreasing (i.e., $y_i\leq z_i\ \forall i\in N\Rightarrow C(y)\leq C(z)$), and that $C(x^0)=0$.
	\end{itemize}
We denote by ${\cal P}^N$ the set of deterministic scheduling problems with delays with player set $N$, and by ${\cal P}$, the set of deterministic scheduling problems with delays.
\end{defi}

Note that the first three items of a deterministic scheduling problem with delays characterise a project. Operational researchers have developed several methodologies for project management. In particular, the minimum duration of a project $(N,\prec,x^0)$, provided that all restrictions imposed by $\prec$ are satisfied, can be obtained as the solution of a linear programming problem, and thus, can be easily computed. We denote the minimum duration of $(N,\prec,x^0)$ by $d(N,\prec,x^0)$. Alternatively, $d(N,\prec,x^0)$ can be calculated using a project planning methodology like PERT (see, for instance, Hillier and Lieberman (2001) for details on project planning). The delay cost function $C$ in Definition \ref{dsp} is rendered in a general way but typically depends on the minimum duration of the project, i.e., $C(y)=c(d(N,\prec ,y))$ for a non-decreasing function $c:\mathbb{R}\rightarrow \mathbb{R}$ with $c(d(N,\prec,x^0))=0$.

In a deterministic scheduling problem with delays $P$, the main question to be answered is how to allocate $C(x)$ amongst the activities in a fair way. This issue has been taken up, for instance, in Berganti\~{n}os et al. (2018); they introduce the Shapley rule in this context.

\begin{defi}
	\label{defdetrul}
	A rule for deterministic scheduling problems with delays is a map $\varphi$ on ${\cal P}$ that assigns to each $P=(N,\prec,x^0,x,C)\in{\cal P}^N$ a vector $\varphi(P)\in\mathbb{R}^N$ satisfying:
	\begin{enumerate}
		\item Efficiency (EFF).
		$\sum_{i\in N} \varphi_i(P)=C(x)$.
		\item Null Delay (ND).
		$\varphi_i(P)=0$ when $x_i=x_i^0$.
	\end{enumerate}
\end{defi}

\begin{defi}
The Shapley rule for deterministic scheduling problems with delays $Sh$ is defined by
$Sh(P)=\varPhi(v^P)$
where for all $P\in {\cal P}^N$ 
\begin{itemize}
	\item 
	$v^P$ is the TU-game with set of players $N$ given by
	$v^P(S)=C(x_S,x^0_{N\setminus S})$	for all $S\subseteq N$ (where $x_S,x^0_{N\setminus S}$ denotes the vector in $\mathbb{R}^N$ whose $i$-th component is $x_i$ if $i\in S$ or $x_i^0$ if $i\in N\setminus S$), and 
	\item 
	$\varPhi(v^P)$ denotes the proposal of the Shapley value for $v^P$. 
\end{itemize}
\end{defi}

For those unfamiliar with cooperative game theory, a TU-game is a pair $(N,v)$ where $N$ is a non-empty finite set, and $v$ is a map from $2^N$ to $\mathbb{R}$ with $v(\emptyset)=0$. We say that $N$ is the player set of the game and $v$ is the characteristic function of the game, and we usually identify $(N,v)$ with its characteristic function $v$. We denote by $G^N$ the set of all TU-games with player set $N$, and by $G$ the set of all TU-games. The Shapley value is a map $\varPhi$ that associates with every TU-game $(N,v)$ a vector $\varPhi (v)\in \mathbb{R}^N$ satisfying $\sum_{i\in N}\varPhi_i (v)=v(N)$ and providing a fair allocation of $v(N)$ to the players in $N$. The explicit formula of the Shapley value for every TU-game $(N,v)$ and every $i\in N$ is given by: 
\[\varPhi_i (v)=\sum _{S\subseteq N\setminus \{i\}}{\frac {(|N|-|S|-1)!\;|S|!}{|N|!}}(v(S\cup \{i\})-v(S)).\]
Since its introduction by Shapley (1953), the Shapley value has proved to be one of the most important rules in cooperative game theory and has applications in many practical problems (see, for instance, Flores et al. (2007)).

Berganti\~{n}os et al. (2018) showed that the Shapley value has good properties in this context and provided an axiomatic characterisation of their Shapley rule by using a consistency property. In this paper, we introduce a generalization of the model and the Shapley rule described above by assuming that the durations of the activities are stochastic. Let us first introduce and motivate interest in our model.

\begin{defi}
	\label{ssp}
	A stochastic scheduling problem with delays $SP$ is a tuple $(N,\prec,X^0,x,C)$ where:
	\begin{itemize}
		\item $N$ is the finite non-empty set of activities.
		\item $\prec$ is a binary relation over $N$ satisfying asymmetry and transitivity. 
		\item $X^0\in\mathbb{R}^N$ is a vector of independent random variables. For every $i\in N$, $X^0_i$ is a non-negative random variable describing the duration of activity $i$.
		\item $x\in\mathbb{R}^N$ is the vector of actual non-negative durations. 
		\item $C:\mathbb{R}^N\rightarrow \mathbb{R}$ is the delay cost function. We assume that $C$ is non-negative and non-decreasing.
	\end{itemize}
	We denote by ${\cal SP}^N$ the set of stochastic scheduling problems with delays with player set $N$, and by ${\cal SP}$ the set of all stochastic scheduling problems with delays.
\end{defi}

Note that in a stochastic scheduling problem with delays, the durations are non-negative random variables instead of non-negative numbers. In general, the duration of an activity can now take any non-negative real value, and  conditions generalising $x_i\geq x^0_i$ or $C(x^0)=0$ as in Definition \ref{dsp} cannot be stated. In the stochastic setting, a delay in an activity is unclear. However, if the actual duration of an activity is longer than the upper bound of its distribution support, it has thus been delayed. Moreover, if its duration is in the 99th percentile of the distribution of its duration, one may think that it has been delayed somewhat. However, what should we think when its actual duration is in the 56th percentile? In the deterministic setting, we can clearly observe when an activity has been delayed. Another novelty in the stochastic setting is that an activity may somehow be delayed, but it may also somehow be ahead of schedule (for instance, when its duration is in the first percentile). In the deterministic setting, by contrast, the case $x_i< x^0_i$ is generally discarded. In any case, although we propose our model in general, our objective is to distribute delay costs when they occur (because $P(x_i\geq X^0_i)$ is large, at least for some $i\in N$), and in situations in which there should not be delays a priori, in the sense that $P(C(X^0)=0)$ is large.

We give next the definition of a rule in this setting. As the meaning of a delay is not clear, this definition does not contain a kind of null delay property, as in Definition \ref{defdetrul}. 

\begin{defi}
	\label{defstochrule}
	A rule for stochastic scheduling problems with delays is a map $\psi$ on ${\cal SP}$ that assigns to each $SP=(N,\prec,X^0,x,C)\in{\cal SP}^N$ a vector $\psi(SP)\in\mathbb{R}^N$ satisfying 
		$\sum_{i\in N} \psi_i(SP)=C(x)$.
\end{defi}

 A first approach to deal with a stochastic scheduling problem with delays is to build from it an associated deterministic problem. More precisely, for a given $SP=(N,\prec,X^0,x,C)\in{\cal SP}^N$, it is natural to associate with it the problem $\overline{SP}=(N,\prec,E(X^0),x,C)$, where $E(X^0)=(E(X^0_i))_{i\in N}$, $E(X^0_i)$ denotes the mathematical expectation of random variable $X_i^0$. This approach encounters a technical obstacle: $\overline{SP}$ is not always a deterministic scheduling problem with delays in the sense of Definition \ref{dsp} because $E(X^0_i)$ can be greater than $x_i$ for some $i$, and $C(E(X^0))$ may be different from zero. This obstacle can be overcome with small adjustments in the definition of an associated deterministic problem. Besides, in many particular examples, we do not encounter this obstacle. In any case, this approach is not the most appropriate because it does not use all the relevant information given in the original problem. Let us illustrate this shortcoming in the following example:

\begin{exa}
	\label{example1}
	Consider the stochastic scheduling problem with delays $SP=(N,\prec,X^0,x,C)$ given by:
	$$\begin{tabular}{|c||c|c|}
	\hline
	$N$&1&2\\\hline\hline
	$\prec$&-&-\\\hline
	$X^0$&$U(0,10)$&$U(2,8)$\\\hline
	$x$&7&7\\\hline
	\end{tabular}$$
	and, for every $y\in\mathbb{R}^N$,
	$$C(y)=\left\{\begin{array}{lc}
	0&\mbox{if $d(N,\prec,y)\leq 6$,}\\
	d(N,\prec,y)-6&\mbox{otherwise.}
	\end{array}\right.$$
	Note that for all $i\in N$ the $i$-th column displays:  
	\begin{itemize}
		\item Activities that precede activity $i$. In this example, $\prec=\emptyset$, i.e., the two activities can be carried out simultaneously. In general, the row corresponding to $\prec$ only shows the immediate precedences, i.e., some elements of $\prec$, but the entire $\prec$ can be easily obtained as the smallest transitive binary relation over $N$ that contains the given elements of $\prec$. An illustration of this can be found in Example \ref{example2}.
		\item The distribution of $X^0_i$. In this case, $X^0_1$ and $X^0_2$ are random variables with a uniform distribution of $U(0,10)$ and $U(2,8)$, respectively.
		\item $x_i$, the duration of $i$; in this case, $x=(7,7)$.
		\end{itemize}  
Note that in this example, $E(X^0_1)=E(X^0_2)=5$, and activities $1$ and $2$ are indistinguishable in $\overline{SP}$. Hence, the anonymity property satisfied by the Shapley rule for deterministic scheduling problems with delays (see Berganti\~{n}os et al. (2018)) implies that $Sh (\overline{SP})=(\frac{1}{2},\frac{1}{2})$. However, activities $1$ and $2$ are distinguishable in $SP$ because the expected duration of the project conditioned to $x_1=7$ is smaller than the expected duration of the project conditioned to $x_2=7$. It seems that a fair rule should take this into account and allocate to activity $2$ a larger part of the delay cost.
\end{exa}

In the next section, we provide a rule for stochastic scheduling problems with delays that overcomes the technical obstacle described above and, more importantly, the drawback described in Example \ref{example1}.

\section{ Shapley rule for stochastic scheduling problems with delays}

In this section, we define and study the Shapley rule for stochastic scheduling problems with delays.

\begin{defi}
	The Shapley rule for stochastic scheduling problems with delays $SSh$ is defined by
	$SSh(SP)=\varPhi(v^{SP})$
	where for all $SP\in {\cal SP}^N$ 
	\begin{itemize}
		\item 
		$v^{SP}$ is the TU-game with set of players $N$ given by
		$v^{SP}(S)=E(C(x_S,X^0_{N\setminus S}))$
		for all non-empty $S\subseteq N$,\footnote{As in all TU-games, we define $v^{SP}(\emptyset )=0$.} and 
		\item 
		$\varPhi(v^{SP})$ denotes the proposal of the Shapley value for $v^{SP}$. 
	\end{itemize}
\end{defi}

This rule inherits many properties of the Shapley value. For instance, it is easy to check that it satisfies the correspondingly modified versions of the properties proved in Berganti\~{n}os et al. (2018) for the Shapley rule for deterministic scheduling problems with delays. In this study, we focus on a different property of the Shapley value and how to adapt it to our context: the balancedness property.

A rule for stochastic scheduling problems with delays satisfies the balancedness property if it treats all pairs of activities in a balanced way, which more precisely means that for every pair of activities $i$ and $j$, the effect of the elimination of $i$ on the allocation to $j$ (according to the rule) is equal to the effect of the elimination of $j$ on the allocation to $i$. To write this property formally, consider a stochastic scheduling problem with delays $SP=(N,\prec,X^0,x,C)\in {\cal SP}^N$, with $|N|\geq 2$, and $i\in N$. Now, we define the resulting problem if activity $i$ is eliminated $SP_{-i}\in {\cal SP}^{N\setminus i}$ by
$$SP_{-i}=(N\setminus i,\prec_{-i},X^0_{-i},x_{-i},C_{-i})$$
where:
\begin{itemize}
	\item $\prec_{-i}$ is the restriction of $\prec$ to $N\setminus i$,
	\item $X^0_{-i}$ is the vector equal to $X^0$ after deleting its $i$-th component,
	\item $x_{-i}$ is the vector equal to $x$ after deleting its $i$-th component, and
	\item $C_{-i}:\mathbb{R}^{N\setminus i}\rightarrow\mathbb{R}$ is given by $C_{-i}(y)=E(C(y, X_i^0))$, for all $y\in\mathbb{R}^{N\setminus i}$.
\end{itemize}
We now formally write the balancedness property.

\vspace*{0.25cm}

\noindent
{\bf Balancedness}. A rule for stochastic scheduling problems with delays $\psi$ satisfies the balancedness property when
$$\psi_i(SP)-\psi_i(SP_{-j})=\psi_j(SP)-\psi_j(SP_{-i})$$
for all $SP\in{\cal SP}^N$, all finite $N$, and all $i,j\in N$ with $i\neq j$.

\vspace*{0.25cm}

The following theorem shows that the balancedness property characterises the Shapley rule.

\begin{theor}
The Shapley rule is the unique rule for stochastic scheduling problems with delays that satisfies the balancedness property.
\end{theor}
\begin{proof}
Let us first check that the Shapley rule satisfies the balancedness property. Take $SP=(N,\prec,X^0,x,C)\in{\cal SP}^N$ and $i,j\in N$ with $i\neq j$. Then,
\begin{equation}
\label{e1}
SSh_i(SP)-SSh_i(SP_{-j})=\Phi_i(v^{SP})-\Phi_i(v^{SP_{-j}}),
\end{equation}
\begin{equation}
\label{e2}
SSh_j(SP)-SSh_j(SP_{-i})=\Phi_j(v^{SP})-\Phi_j(v^{SP_{-i}}).
\end{equation}
Now, for every $k\in N$, $v^{SP_{-k}}$ is a TU-game with set of players $N\setminus k$. For every non-empty $S\subseteq N\setminus k$, \footnote{To facilitate the reading of this proof, when dealing with the mathematical expectation of a random vector, we explicitly indicate the components of the vector to which the mathematical expectation refers.}
\begin{eqnarray}
\nonumber
v^{SP_{-k}}(S)&=&E_{N\setminus (S\cup k)}(C_{-k}(x_S,X^0_{N\setminus (S\cup k)}))\\
\nonumber
&=&E_{N\setminus (S\cup k)}(E_k(C(x_S,X^0_{N\setminus (S\cup k)}, X^0_k)))
\end{eqnarray}
Now, the independence of the components of $X^0$ implies that
\begin{equation}
\nonumber
v^{SP_{-k}}(S)=E_{N\setminus S}(C(x_S,X^0_{N\setminus S}))=v^{SP}(S).
\end{equation}
Note that for every $S\subseteq N\setminus k$, $v^{SP}(S)=v_{-k}^{SP}(S)$, where $v_{-k}^{SP}\in G^{N\setminus\{ k\}}$ denotes the restriction of the TU-game $v^{SP}\in G^N$ to $N\setminus\{ k\}$. Hence,
\begin{equation}
\label{e3}
v^{SP_{-k}}=v^{SP}_{-k} \mbox{ for all $k\in N$}.
\end{equation}
Considering (\ref{e3}) and that Myerson (1980) proved that the Shapley value of a TU-game satisfies a balancedness property, the equations in (\ref{e1}) and (\ref{e2}) are equal. This implies that the Shapley rule satisfies the balancedness property.

Suppose now that there exists another rule $R\neq SSh$ for stochastic scheduling problems with delays that satisfies the balancedness property. As $R\neq SSh$, there must exist $SP=(N,\prec,X^0,x,C)\in{\cal SP}$ with $R(SP)\neq SSh(SP)$. Assume that $SP$ is minimal, in the sense that: (a) $|N|=1$, or (b) $|N|\geq 2$ and $R(SP_{-i})= SSh(SP_{-i})$ for every $i\in N$.\footnote{ This assumption is without loss of generality because if $SP\in {\cal SP}^N$ is not minimal, we can eliminate one by one the elements of $N$ until we have a minimal $SP'$ with $R(SP' )\neq SSh(SP' )$.} Note that $|N|\neq 1$ because otherwise, $R(SP)=C(x)= SSh(SP)$; hence, $|N|\geq 2$. Take $i,j\in N$ with $i\neq j$. As $R$ and $SSh$ satisfy the balancedness property, then
$$R_i(SP)-R_j(SP)=R_i(SP_{-j})-R_j(SP_{-i}),$$
$$SSh_i(SP)-SSh_j(SP)=SSh_i(SP_{-j})-SSh_j(SP_{-i}).$$
Now, considering the minimality of $SP$, 
$$R_i(SP)-R_j(SP)=SSh_i(SP)-SSh_j(SP)$$
or, equivalently, $R_i(SP)-SSh_i(SP)=A\in\mathbb{R}$, i.e. it does not depend on $i$. But then, $A=0$ because $\sum_{j\in N} R_j(SP)=C(x)=\sum_{j\in N}SS_j(SP)$. This implies that $R(SP)= SSh(SP)$, and the proof is concluded.
\end{proof}

In the remainder of this section, we illustrate the performance of the Shapley rule in two examples. 

Note first that the Shapley rule behaves in Example \ref{example1} as desired. For the stochastic scheduling problem with delays $SP$, we can easily check that:
\begin{itemize}
	\item $v^{SP}(1)=E(C(7,X_2^0))=13/12$,
	\item $v^{SP}(2)=E(C(X_1^0,7))=29/20$,
	\item $v^{SP}(N)=C(7,7)=1$,
\end{itemize}
and then, $SSh(SP)=(0.31666, 0.68333)$. Thus, activity $2$ receives a larger part of the delay cost, as it should. Note that in this example, $SSh(SP)$ can be easily exactly calculated. In general, $SSh$ cannot be exactly calculated, but can be estimated using simulation techniques. Consider now a new example that is slightly more complex.

\begin{exa}
	\label{example2}
	Consider the stochastic scheduling problem with delays $SP=(N,\prec,X^0,x,C)$ given by:	
	$$\begin{tabular}{|c||c|c|c|c|c|}
	\hline
	$N$&1&2&3&4&5\\\hline\hline
	$\prec$&-&1&-&1,3&2\\\hline
	$X^0$&t(1,2,3)&t(1/2,1,3/2)&t(1/4,1/2,9/4)&t(3,4,5)&exp(1/2)\\\hline
	$x$&2.5&1.25&2&4.5&3\\\hline
	\end{tabular}$$
	and, for every $y\in\mathbb{R}^N$,
	$$C(y)=\left\{\begin{array}{lc}
	0&\mbox{if $d(N,\prec,y)\leq 6.5$,}\\
	d(N,\prec,y)-6.5&\mbox{otherwise,}
	\end{array}\right.$$
	where $t(a,b,c)$ denotes the triangular distribution with parameters $a$, $b$, and $c$, and exp($\alpha $) denotes the exponential distribution with parameter $\alpha$. As we remarked in Example \ref{example1}, the table does not give the entire binary relation $\prec$ but only the immediate precedences. For instance, because 1 precedes 2, 2 precedes 5 and $\prec$ is transitive, then 1 must precede 5; however, the table only indicates that 2 precedes 5. The entire $\prec$ is easily obtained as the smallest transitive binary relation over $N$ that contains the given elements of $\prec$. In this case, the table displays
	$$(1,2), (1,4), (3,4), (2,5)$$
	and then 
	$$\prec=\{(1,2), (1,4), (1,5), (3,4), (2,5)\}.$$
	In some cases, it is more instructive to give the PERT graph representing the precedences instead of the precedences and $\prec$. The PERT graph in this example is given in Figure \ref{figure_1}, where, for each arc, we indicate the activity that it represents and the duration of this activity according to $x$. 
	\begin{figure}
	\centerline{\includegraphics[scale=0.75]{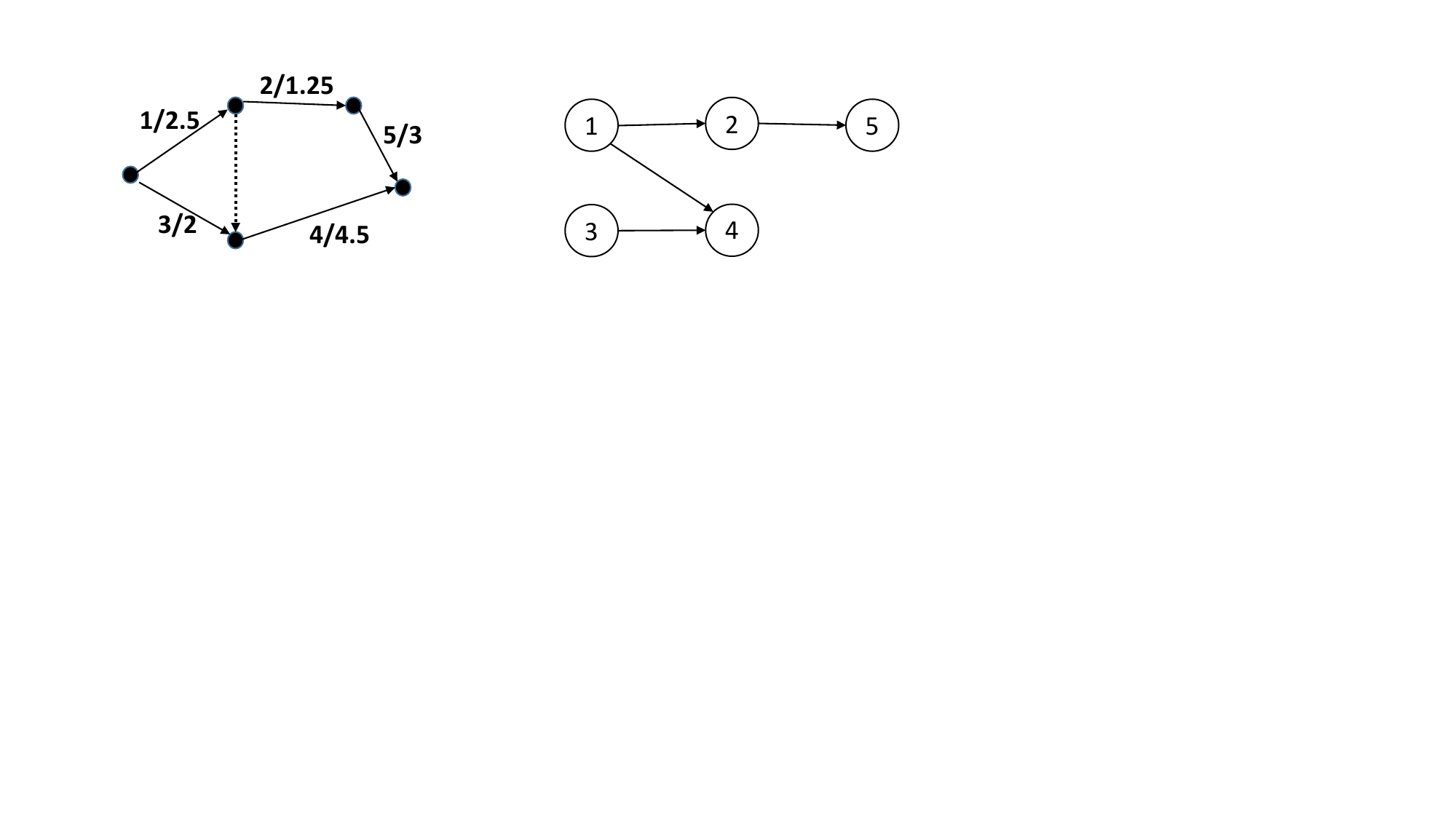}}
	\caption{PERT graph of the project in Example \ref{example2}}
	\label{figure_1}
	\end{figure}
	It is easy to check that $d(N,\prec,x)=7$ (remember that the duration of a project is equal the duration of its longest path in the PERT graph), and then $C(x)=0.5$. To allocate this cost amongst the activities in a fair way, note first that $E(X^0)=(2,1,1,4,2)$, and thus, all activities have a delay with respect to their expected durations. If we take the naive approach, we can allocate the delay cost by using the Shapley rule for $\overline{SP}=(N,\prec,E(X^0),x,C)$. In this case,
	$$Sh(\overline{SP})=(0.27083, 0.02083, 0, 0.18750, 0.02083).$$
	At first sight, this is a reasonable allocation of the delay cost. Activities 1 and 4 belong to the longest path in project $(N,\prec,x)$, and thus, receive most of the delay cost. The cost allocated to activity 1 is greater than that allocated to activity 4 because activity 1 also belongs to a path with a duration greater than $6.5$ (the path 1-2-5 has duration $6.75$). Activity 3 only belongs to one path with duration $6.5$, and produces no delay cost. Therefore, it pays $0$. However, note that this allocation does not consider the probability distributions of the durations of the activities but only their averages. For instance, the duration of activity 5 follows an exponential distribution, the support for which is $[0,\infty)$. This means that its duration can be very long, and therefore, can produce a longer delay. However, its duration is not very long; so, in a sense, activity 5 contributes to a lack of delay in the project. This is captured by the Shapley rule for stochastic scheduling problems with delays. Using elementary simulation techniques, $SSh(SP)$ can be estimated and the result is
	$$SSh(SP)=(0.28960,0.09834,0.07641,0.20659,-0.17095).$$
	It should be noted that this allocation differs from $Sh(\overline{SP})$ primarily in that activity 5 receives a kind of reward for not being too late, where this reward is paid by activities 1, 2, and 4, which last longer than expected and belong to paths whose durations entail a delay cost. 
	
	We now use a small simulation experiment indicating that, on the average, when $x$ is drawn from $X^0$, the cost allocation provided by SSh causes activity 5 to pay the largest part of the delay cost. We then realise that SSh tends to allocate the delay cost to activities 1, 4, and 5, but that it is very sensitive to the durations of the activities. We simulated 1,000 times the durations of the activities such that the 1,000 corresponding durations of the projects were greater than $6.5$, i.e. we simulated $(x^i)_{i\in\{ 1,\dots ,1,000\}}$, each $x^i_j$ being an  observation of $X^0_j$, all drawn independently and in such a way that $C(x^i)>0$. Thus, we obtained 1,000 stochastic scheduling problems with delays $SP^i=(N,\prec,X^0,x^i,C)$ as well as their 1,000 associated proposals of the Shapley rule $SSh(SP^i)$. We then calculated
	\begin{equation}
	\label{e4}
	\sum_{i\in\{ 1,\dots ,1000\}}\frac{SSh(SP^i)}{1000}=(0.12857, 0.06844, 0.06686, 0.10757, 0.93790),
	\end{equation}
	where the average observed cost was $1.30935$. Note that (\ref{e4}) showed that, in effect, when there are positive delay costs in an implementation of the stochastic project $SP=(N,\prec,X^0)$ the delay cost function being $C$, the cost allocation provided by $SSh$ primarily burdens activity 5. This suggests that the vector of actual durations $x$ that we handle in this example could be considered atypical because $SSh_5(SP)<0$.
	Figure \ref {figure_2} confirms it. It displays the density estimations of the variables $Z^1_i$ (solid line) and $Z^2_i$ (dotted line), $i\in\{ 1,\dots 5\}$, such that
	\begin{itemize}
		\item $Z^1_i$ is the $i$-th component of $Sh((N,\prec,E(X^0),X,C))$, where $X$ denotes the random variable corresponding to an observation of $X^0$; and
		\item $Z^2_i$ is the $i$-th component of $SSh((N,\prec,X^0,X,C))$, where $X$ denotes the random variable corresponding to an observation of $X^0$.
	\end{itemize} 
Note that the scales of the five graphics in Figure \ref {figure_2} are different, which is a relevant feature to interpret them. It is not possible to adjust the scales while maintaining the informative graphics. From the figure, we see that the probability that $Sh_5((N,\prec,E(X^0),X,C))<0$ is not high. It is interesting to note that the variables $Z^1_i$ and $Z^2_i$ are significantly different for each $i$, which strengthens the interest of the rule $SSh$.

Nota Tabla \ref{pos.neg}: En esta red PERT, hay $3$ caminos. El camino dado por las actividades $1-2-5$, $1-4$ y $3-4$. El camino $3-4$ nunca genera retraso. El camino $1-4$ es aquel que más veces genera retraso, y el $1-2-5$, dado que $5$ es una exponencial, cuando genera retraso es muy elevado. Lo que sucede en el caso determinista es que en el camino $1-2-5$, si la actividad $5$ no se retrasa este camino tampoco lo hace provocando que las tres actividades reciban un pago de $0$. En cambio si la actividad $5$ toma un valor alto, esta recibe un pago positivo y las actividades $1$ y $2$ pueden recibir tanto un pago positivo o negativo dependiendo de sus realizaciones. Por tanto, en el caso determinista, $5$ nunca recibe un pago negativo. En el caso estocástico, se puede apreciar que el $48.3\%$ de las veces la actividad $5$ toma un valor suficientemente alto para generar retraso y recibir un pago positivo, en cambio, el $51.7\%$ no se retrasa, y siguiendo el contexto determinista recibiría un pago de $0$, pero al considerar la variabilidad se tiene en cuenta su distribución y en como podría afectar negativamente al proyecto, generando un retraso elevado, y no lo hace obteniendo un pago negativo.
\begin{table}[htbp]
	\begin{minipage}[b]{0.5\linewidth}

		\begin{center}
			\begin{tabular}{|c||c|c|c|c|c|}
				\hline
				Sh& $1$& $2$& $3$& $4$& $5$ \\
				\hline \hline
				$\geq 0$ & $70.5$ & $74.9$& $100$& $91.1$& $100$\\ \hline
				$<0$ & $29.5$& $25.1$& $0.0$& $8.9$& $0$ \\ \hline
			\end{tabular}
		
		\end{center}

\end{minipage}
	\begin{minipage}[b]{0.5\linewidth}
	\begin{center}
		\begin{tabular}{|c||c|c|c|c|c|}
			\hline
			SSh& $1$& $2$& $3$& $4$& $5$ \\
			\hline \hline
			$\geq 0$ & $75.0$ & $85.5$& $100$& $95.5$& $48.3$\\ \hline
			$<0$ & $25$& $14.5$& $0.0$& $4.5$& $51.7$ \\ \hline
		\end{tabular}
	
	\end{center}
\end{minipage}

\caption{Positive and negative payments for the Sh rule (left) and SSh rule (right).}
\label{pos.neg}
\end{table}

	\begin{figure}
	\includegraphics[scale=0.5]{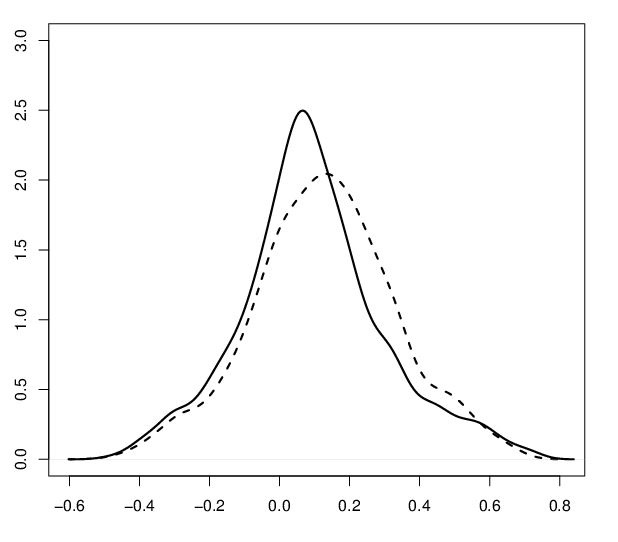}
	\includegraphics[scale=0.5]{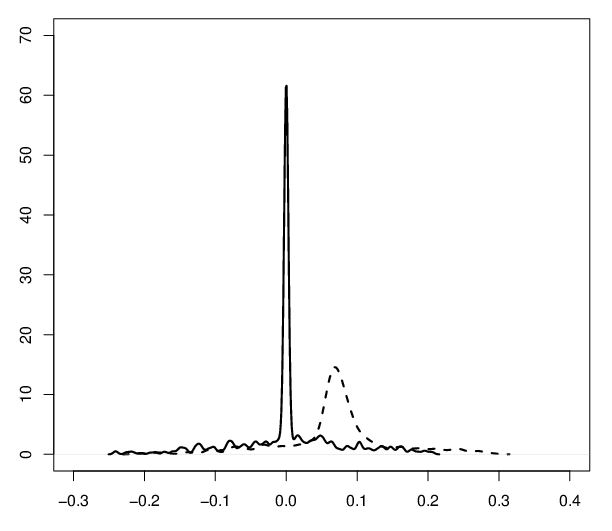}
	\includegraphics[scale=0.5]{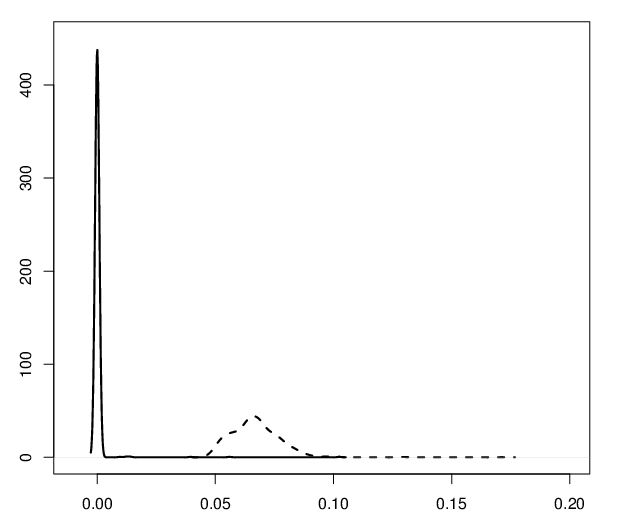}
	\includegraphics[scale=0.5]{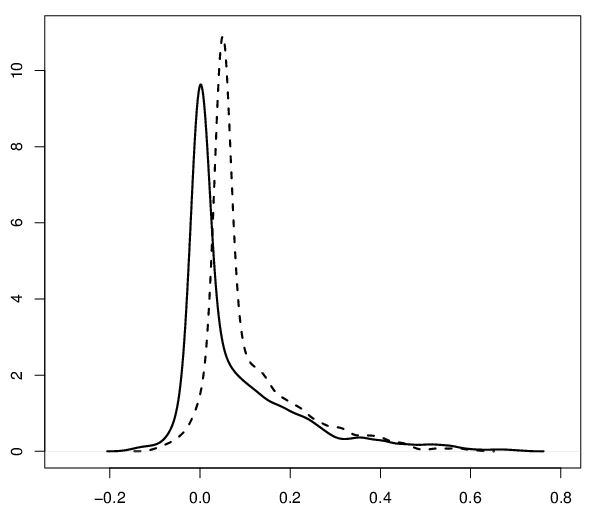}
	\includegraphics[scale=0.5]{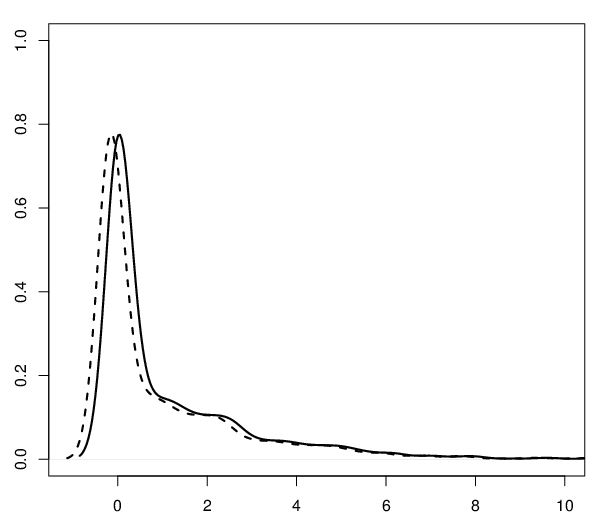}
		\caption{Density estimations of the variables $Z^1_i$ (solid line) and $Z^2_i$ (dotted line)}
		\label{figure_2}
	\end{figure}
\end{exa}
		
\section{Computational Analysis}

The calculation of the Shapley value has, in general, an exponential complexity. Although equivalent expressions with polynomial complexity can be used in some game classes, this is not the case for the class of games with which we are dealing. Calculating the Shapley value in our context is impossible in practise, even for a moderate number of activities. For example, if the number of activities is $100$, there are $2^{100}$ coalitions in which the characteristic function must be evaluated. As an alternative to exact calculation, Castro et al. (2009) proposed an estimate of the Shapley value in polynomial time using a sampling process. We also need to calculate $v^{SP}(S)=E(C(x_S,X^0_{N\setminus S}))$, with $S\subseteq N$. In some simple cases, this value can be calculated in a simple way using the properties of order statistics; but in general, we use simulations to approximate $v^{SP}$.

The aims of this section are twofold: First, to illustrate the implementation of the computation of the Shapley rule through its pseudocode, from which it is easy to check that the computational complexity of our rule is $O(n^4)$; and second, to show by examples that it is possible to approximate the Shapley rule for stochastic scheduling problems with delays in an acceptable time, even if there are hundreds of activities, by using a desktop computer and free software. The error in the two phases of estimation is tracked a posteriori through the estimation of variance and central limit theorem.
	
The first task of code is to reorder the precedence matrix: That is to say, if the value at $(i,j)$ is equal to $1$, it means that $i$ precedes $j$, and we want $i< j$.  Note that this task can always be carried out and allows for faster calculation. We denote the $i$-th row of matrix $P$ by $P_{i,\cdot}$ and $i$-th column by $P_{\cdot,i}$.

\vspace{0.25cm}

\textbf{Organise precedence matrix}
\rule[1mm]{60mm}{0.1mm}
\begin{itemize}
	\item \textbf{Begin}
	\subitem $P=precedence$, $index=NULL$
	\subitem \textbf{While} number of P's columns $>0$
	\subsubitem Take all $i\in n$ such that $ \sum_{j=1}^{n}P_{ij}=0$
	\subsubitem $index=(index,i)$
	\subsubitem $P=P\backslash P_{i,\cdot}$ and $P=P\backslash P_{\cdot,i}$
	\subitem \textbf{end}
	\subitem $precedence=precedence_{index,index}$
	\item \textbf{end}
\end{itemize}

\rule[4mm]{100mm}{0.1mm}

The code computes the early times for a deterministic scheduling problem when the duration of the activities is given by $x^{0}$. The early time of an activity is the earliest that this activity can begin.

\vspace{0.25cm}

\textbf{Early times}
\rule[1mm]{80mm}{0.1mm}
\begin{itemize}
	\item \textbf{Begin}
	\subitem $early.times_{i}=0$ $\forall i\in N$
	\subitem \textbf{Organise} precedence matrix
	\subitem $I=\{i\in n,$ such that $ \sum_{j=1}^{n}precedence_{ji}\neq0\}$
	\subitem \textbf{For} $i\in I$
	\subsubitem  $prec=\{j\in n/precedence_{ji}=1\}$
	\subsubitem  $early.times_{i}=\max\{x^{0}_{prec}+early.times_{prec}\}$
	\subitem \textbf{end}
	\item \textbf{end}
\end{itemize}

\rule[4mm]{100mm}{0.1mm}

Let us consider a deterministic scheduling problem with delays with delay cost function, for every $y\in\mathbb{R}^N$, given by:
$$C(y)=\left\{\begin{array}{lc}
0&\mbox{if $d(N,\prec,y)\leq\delta$,}\\
d(N,\prec,y)-\delta&\mbox{otherwise.}
\end{array}\right.$$
We obtain an estimation of the Shapley rule in polynomial time. The algorithm consists of taking $m\in\mathbb{N}$ permutations of the set of players $N$ with equal probability (Castro et al., 2009). We denote by $\Pi_{N}$ the set of permutations of $N$. We then calculate $|N|$ real numbers as follows:
\[\pi^{j} \in \Pi_{N}\ \ \text{where}\ \ \pi^j=(\pi^j_{1},...,\pi^j_{|N|})\ \ \text{and}\ \ j\in\{1,...,m\} \]
\[x(\pi^j)_{i}=v(Pre^{i}(\pi^j)\cup\{i\})-v(Pre^{i}(\pi^j))\]
where $Pre^{i}(\pi^j)=\{\pi^j_{z}\in\pi^j; z<i\}$; $x(\pi^j)\in \mathbb{R}^{|N|}$ is the corresponding allocation vector. Finally, the estimated value of the Shapley value is:
\[\hat{Sh}_i=\frac{1}{m}\sum_{j\in 1}^{m}x(\pi^j)_{i}\]
for all $i \in N$.

When we address the stochastic version of the problem, we can use nearly the identical procedure to that in the deterministic case; but in this new situation, we need to estimate the TU-game. For this, we simulate the TU-game $m_1\in\mathbb{N}$ times and take the average of these values.

\vspace{0.25cm}

\textbf{Estimation of Shapley rule in the stochastic case}
\rule[1mm]{10mm}{0.1mm}
\begin{itemize}
	\item \textbf{Begin}
	\subitem Determine $m$ and $m_1$
	\subitem $Cont=0$, $\hat{Sh}_{i}=0$, $v_{i}=0$ $\forall i\in N$ and $time_{j}=0$ $\forall j\in m_1$
	\subitem \textbf{For} $j\in m_1$
	\subsubitem $\hat{X}^0_{j,\cdot}=sample(X^0)$
	\subitem end
	\subitem \textbf{Organise} precedence matrix
	\subitem \textbf{While} $cont<m$
	\subsubitem Take $\pi \in \Pi_{N}$ with probability $\frac{1}{n!}$
	\subsubitem \textbf{For} $i\in n$
	\subsubitem \ \ \textbf{For} $j\in m_1$
	\subsubitem \ \ \ \ \textbf{Early times} of $(x_{Pre^{i}(\pi)\cup \{i\}},\hat{X}^0_{j, N\setminus \{Pre^{i}(\pi)\cup \{i\}\}})$
	\subsubitem \ \ \ \ $\displaystyle v_{j}=\max\{\max\{early.times+(x_{Pre^{i}(\pi)\cup \{i\}},\hat{X}^0_{j,N\setminus \{Pre^{i}(\pi)\cup \{i\}\}})\}-\delta,0\}$ 
	\subsubitem  \ \ end
	\subsubitem \ \ $v_i=mean(time)$
	\subsubitem end
	\subsubitem $\hat{Sh}_{\pi_{1}}=\hat{Sh}_{\pi_{1}}+v_{1}$
	\subsubitem $\hat{Sh}_{\pi_{i}}=\hat{Sh}_{\pi_{i}}+v_{i}-v_{i-1}$ $\forall i\in N\backslash\{1\}$
	\subsubitem $cont=cont +1$
	\subitem end
	\subitem $\hat{Sh}=\frac{\hat{Sh}}{m}$
	\item \textbf{end}
\end{itemize}

\rule[4mm]{100mm}{0.1mm}

To gain insight into the computation time needed to obtain a solution, we selected five problems\footnote{These problems were too large to be included in this paper. They can be downloaded from http://dm.udc.es/profesores/ignacio/stochasticprojects} with a number of activities ranging from 10 to 1,000. We ran the problems on a PC with a $3.70$ GHz Core i$7$-$8700$K, and $64$ GB of RAM on an Ubuntu $64$-bits. The programming language used was \textbf{R} x$64$ $3.4.4$. It is freely available under the GNU General Public License. To improve performance in terms of time, we used the packages \textbf{Rcpp} and \textbf{parallel}. The package Rcpp was used to write in C the function \textit{early times} and parallel was used to parallelise the estimation of the Shapley value by using six cores of our computer.

Table \ref{tiempos.simu} shows the computation times, in seconds, of the five problems, with 10, 30, 100, 300, and 1,000 activities, respectively. The TU-game was approximated using $m_1=1000$ simulations, while $m=1000$ and $10000$ estimates were used for the Shapley rule.

\begin{table}[htbp]
	\begin{center}
		\begin{tabular}{|c||c|c|c|c|c|}
			\hline
			& $10$& $30$& $100$& $300$& $1000$ \\
			\hline \hline
			$1000$ & $18$ & $120$& $1033$& $7801$& $118770$\\ \hline
			$10000$ & $211$& $1329$& $11941$& $80521$& $1277377$ \\ \hline
		\end{tabular}
		\caption{Computation times in seconds}
		\label{tiempos.simu}
	\end{center}
\end{table}

\begin{table}[htbp]
	\begin{center}
		\begin{tabular}{|c||c|c|c|c|c|}
			\hline
			& $10$& $30$& $100$& $300$& $1000$ \\
			\hline
			\hline
			$v(S)$ & $2.18$ & $2.96$& $4.64$& $2.28$& $0.83$\\ \hline
			$1000$ & $12.92$& $13.49$& $19.37$& $27.88$& $12.92$ \\ \hline
			$10000$ & $4.17$& $4.27$& $6.13$& $8.82$& $4.09$ \\ \hline
		\end{tabular}
		\caption{Errors for $v(S)$ and the Shapley rule}
		\label{errores.simu}
	\end{center}
\end{table}

Table \ref{errores.simu} shows an estimation of errors, both in the approximation of the characteristic function and Shapley rule by using $m=1000$ and $10000$. All errors are relative and in percent.\footnote{As is common in statistical methodology, the relative error in percent of the estimation of a parameter $\theta$ is given by
	$z_{\alpha/2}\frac{{s}}{\sqrt{n}}\frac{100}{\theta}$, 
	where $s$ is the square root of the sample variance.} A significance level of $\alpha =0.05$ was used in these estimates. The error in $v(S)$ is different for every $S\subseteq N$, and therefore, we display the average of 1,000 coalitions chosen in a random way. In the Shapley rule, each activity has an error, and the table shows the average of all activities.

\section*{Acknowledgements}
This work has been supported by the MINECO grants MTM2014-53395-C3-1-P and MTM2017-87197-C3-1-P, and by the Xunta de Galicia through the ERDF (Grupos de Referencia Competitiva ED431C-2016-015 and Centro Singular de Investigaci\'on de Galicia ED431G/01).

\section*{Compliance with ethical standards}

\textbf{Conflict of interest} There is no potential conflicts of interest.\\
\textbf{Ethical standard} Research do not have human participants and/or animals.

\section*{References}
Berganti\~{n}os G, S\'anchez E (2002). How to distribute costs associated with a delayed project. Annals of Operations Research 109, 159-174.\\
Berganti\~{n}os G, Valencia-Toledo A, Vidal-Puga J (2018). Hart and Mas-Colell consistency in PERT problems. Discrete Applied Mathematics 243, 11-20.\\
Br\^anzei R, Ferrari G, Fragnelli V, Tijs S (2002). Two approaches to the problem of sharing delay costs in joint projects. Annals of Operations Research 109, 359-374.\\
Castro J, G\'omez D, Tejada J (2007). A project game for PERT networks. Operations Research Letters 35, 791-798.\\
Castro J, G\'omez D, Tejada J (2009). Polynomial calculation of the Shapley value based on sampling. Computers \& Operations Research 36, 1726-1730.\\
Castro J, G\'omez D, Tejada J (2014). Allocating slacks in stochastic PERT networks. CEJOR 22, 37-52.\\
Herroelen W, Leus R (2005). Project scheduling under uncertainty: survey and research potentials. European Journal of Operational Research 165, 289-306.\\
Hillier FS, Liberman GJ (2001). Introduction to Operations Research. McGraw-Hill.\\
Est\'evez-Fern\'andez A (2012). A game theoretical approach to sharing penalties and rewards in projects. European Journal of Operational Research 216, 647-657.\\
Est\'evez-Fern\'andez A, Borm P, Hamers H (2007). Project games. International Journal of Game Theory 36, 149-176.\\
Flores, R, Molina, E, and Tejada, J (2007). Evaluating groups with the generalized Shapley value. 4OR, 1-32.\\
Myerson RB (1980). Conference structures and fair allocation rules. International Journal of Game Theory 9, 169-182.\\
Shapley LS (1953) A value for n-person games. In: Kuhn HW, Tucker AW (eds) Contributions to the Theory of Games II. Princeton University Press, Princeton
\end{document}